\documentclass[runningheads]{llncs}
\usepackage[utf8]{inputenc}
\usepackage{amsmath,amssymb,amsfonts,array,booktabs,multirow,makecell,wasysym,pifont,threeparttable}
\usepackage{algorithmic}
\usepackage{graphicx}
\usepackage{textcomp}
\usepackage{xcolor}
\usepackage{stmaryrd}

\usepackage{placeins}
\usepackage[ruled,vlined,linesnumbered]{algorithm2e}
\SetKwInOut{KwIn}{Input} % 改写输入/输出关键字
\SetKwInOut{KwOut}{Output}
\SetKw{KwLet}{Let}
\SetKwBlock{Client}{\textbf{Client-side:}}{}
\SetKwBlock{Server}{\textbf{Server-side:}}{}
\SetKwBlock{Init}{\textbf{Initialization:}}{}
\SetKwFunction{HSDiag}{HSDiag}

\definecolor{AlgBlue}{RGB}{60,120,200}
\newcommand{\act}{\sigma}

\newcommand{\algcommentsty}[1]{\textcolor{AlgBlue}{\itshape #1}}
\SetCommentSty{algcommentsty}

\newcommand{\AlgPhase}[1]{\textcolor{AlgBlue}{\itshape $\triangleright$~#1~$\triangleleft$}\;}

\newcommand{\AlgCmt}[1]{\tcp*[f]{\algcommentsty{#1}}}

\definecolor{BestPink}{RGB}{255,210,210}
\definecolor{SecondYellow}{RGB}{255,242,190}
\newcommand{\best}[1]{\begingroup\setlength{\fboxsep}{1pt}\colorbox{BestPink}{\strut #1}\endgroup}
\newcommand{\second}[1]{\begingroup\setlength{\fboxsep}{1pt}\colorbox{SecondYellow}{\strut #1}\endgroup}

\usepackage{tikz}    % 画竖条

\usepackage[colorlinks,
            linkcolor=blue,
            anchorcolor=blue,
            citecolor=blue,
            urlcolor=blue]{hyperref}

\def\BibTeX{{\rm B\kern-.05em{\sc i\kern-.025em b}\kern-.08em
    T\kern-.1667em\lower.7ex\hbox{E}\kern-.125emX}}

\begin{document}

\title{FBAN: A Fully Homomorphic Encryption Compatible Bottleneck Attention Network for Privacy-Preserving Behavioral Authentication\thanks{This paper has been accepted at the 28th International Conference on Information and Communications Security (ICICS 2026).}}

% \author{Anonymous}

\author{Jichao Xiong\inst{1} \and
 Atsuko Miyaji\inst{1}\thanks{Corresponding author.}\and
 Jiageng Chen\inst{2}
}

% \authorrunning{Anonymous et al.}
\titlerunning{FBAN: FHE-Compatible Bottleneck Attention Network}

\institute{Department of Information and Communications Technology, The University of Osaka, Suita, Osaka 565-0871, Japan \and
School of Computer Science, Central China Normal University, Wuhan 430079, China
\\
\email{\{miyaji\}@comm.eng.osaka-u.ac.jp}}

% \institute{Anonymous}

% \author{Anonymous}

% Use an envelope for the corresponding-author footnote on the title page.
\makeatletter
\let\originalfnsymbol\@fnsymbol
\renewcommand{\@fnsymbol}[1]{%
  \ifnum#1=2\relax\ding{41}\else\originalfnsymbol{#1}\fi}
\makeatother
\maketitle
\makeatletter
\let\@fnsymbol\originalfnsymbol
\makeatother

% Keep the introduction on page one after adding the acceptance footnote.
\vspace{-\baselineskip}
\begin{abstract}
Continuous authentication (CA) strengthens session security by repeatedly verifying the user during device interaction, yet it inherently relies on highly sensitive behavioral traces (e.g., fine-grained touch dynamics) that are often outsourced to cloud/edge services for scalable inference. This raises a fundamental \emph{privacy-in-use} challenge: protecting behavioral features \emph{during} computation, not only in transit or at rest. Fully homomorphic encryption (FHE) offers a principled solution, but deploying modern CA models under FHE remains difficult due to non-linearities and attention-style operations that incur high ciphertext cost.
We propose \textbf{FBAN}, a TFHE-compatible Bottleneck Attention Network and an end-to-end encrypted CA framework. FBAN is designed for integer-only execution via a two-stage pipeline (floating-point pretraining followed by quantization-aware training) and is compiled into TFHE circuits for homomorphic inference. We further specify a client--server protocol with session-bound blinding and decrypt-and-return verification, enabling the server to authenticate users without observing raw behavioral features. We provide a cryptographic security analysis against an honest-but-curious server under TFHE IND-CPA security, and formalize resistance to replay and impersonation without the TFHE secret key. Experiments on two public touchscreen datasets demonstrate that FBAN achieves strong authentication utility under encrypted inference while maintaining a lightweight model footprint, with TFHE parameters instantiated at \(\ge\!128\)-bit security.

\end{abstract}

\section{Introduction}
Smartphones and wearables increasingly mediate security-critical actions—mobile payments, healthcare access, enterprise messaging, and location-dependent services—making user authentication a first-order security primitive~\cite{georgiev2022techniques}. Traditional mechanisms, including passwords/PINs and physiological biometrics (e.g., face or fingerprint), are widely deployed but remain vulnerable to credential theft and presentation or replay attacks, and—crucially—only verify the user at session entry~\cite{abuhamad2020sensor}. Once a session is unlocked, an adversary who gains short-term access can operate under the legitimate identity until the session terminates~\cite{bai2023cryptomask}.

Continuous authentication (CA) has emerged as a promising complement to entry-point authentication by continuously verifying the user in a passive and transparent manner during device interaction~\cite{sitova2015hmog}. Recent CA systems leverage behavioral biometrics such as keystroke dynamics~\cite{yao2026strokepin}, touch gestures, and inertial signals, often augmented by context or temporal modeling~\cite{stylios2021behavioral}; complementary schemes exploit richer physiological cues such as auditory-pupillary response for emerging platforms~\cite{zhu2023soundlock}. For example, touchscreen-based approaches model fine-grained touchpoints and contextual dependencies across consecutive swipes to improve recognition performance while controlling model complexity. Transformer-based frameworks further exploit attention to extract discriminative spatio-temporal representations from multimodal behavioral streams, achieving strong error rates on public datasets~\cite{TSCA}~\cite{dilshan2023behaveformer}.

Despite this progress, a critical gap remains underexplored in most research prototypes: the behavioral and contextual signals enabling CA are intrinsically sensitive personal data~\cite{ometov2018multi}. CA continuously captures and processes fine-grained behavioral traces (e.g., timing, pressure, motion, and usage context), which can reveal not only identity but also private attributes and usage patterns. CA surveys explicitly recognize that such continuous monitoring raises privacy concerns because personal data are collected, transmitted, and processed outside the user’s direct control, particularly in cloud-assisted settings~\cite{ayeswarya2024comprehensive}. In practice, many deployed applications already rely on cloud resources for scalability and on-demand computation~\cite{xiong2024privacy}; however, outsourcing CA inference to an untrusted server exacerbates privacy risk: the server can observe raw behavioral features, intermediate representations, or authentication scores, enabling profiling or unintended secondary use~\cite{baig2021security}.

A straightforward response---running CA entirely on-device---does not fully resolve the problem. Existing continuous authentication (CA) systems largely focus on recognition accuracy and efficiency, while overlooking a fundamental issue: the behavioral stream itself is highly sensitive and is typically processed in the clear. This creates a clear gap between current CA designs and practical privacy requirements. Meanwhile, modern CA models increasingly employ deep architectures with temporal modeling and attention-style components, whose computation, memory, and energy costs can exceed the budget of mobile devices, thereby motivating inference offloading to the edge or cloud~\cite{yang2012novel}. As a result, trustworthy CA requires \emph{privacy-in-use}, i.e., protection of behavioral features throughout inference rather than only in transit or at rest.

Fully homomorphic encryption (FHE) offers a principled way to realize privacy-in-use by enabling computation directly over encrypted data~\cite{xiong2024enhancing}. However, bringing FHE into modern CA remains difficult. On the one hand, nonlinear layers~\cite{gilad2016cryptonets}, normalization~\cite{zhang2024secure}, and especially attention-style interactions~\cite{park2024powerformer} can significantly increase circuit depth and ciphertext cost, leading to substantial latency~\cite{lin2025lightweight}; recent work on quantized attention under FHE for health sensing~\cite{xiong2025stresssentry} further highlights the tension between model expressiveness and encrypted-inference efficiency. On the other hand, typical CA models are designed in floating-point arithmetic with standard activations, whereas efficient encrypted deployment often favors tightly controlled quantized computation~\cite{chen2022x}. Although approximate-arithmetic schemes such as CKKS~\cite{CKKS} support real-valued operations, they are less suitable for our setting, where CA is repeatedly executed over many short behavioral windows and requires stable threshold-based decisions; in this regime, repeated approximate nonlinear evaluation and precision management can still incur expensive refresh operations. We therefore target TFHE~\cite{chillotti2020tfhe}, which naturally supports exact integer-domain computation after quantization and better matches the discrete inference-and-verification workflow considered here.

To address the above challenges---namely, (i) the lack of privacy protection in existing CA despite the sensitivity of behavioral data, (ii) the high homomorphic cost induced by deep nonlinear and attention-based CA models, and (iii) the mismatch between floating-point model design and efficient encrypted execution---we propose \textbf{FBAN}, an FHE-compatible Bottleneck Attention Network for privacy-preserving behavioral continuous authentication. FBAN is designed from the ground up for TFHE-style integer-domain execution: we introduce a bottlenecked serial attention formulation that avoids FHE-unfriendly operators and can be compiled into efficient encrypted circuits after quantization-aware training (QAT). Our end-to-end workflow supports (i) floating-point pretraining for strong baseline accuracy, (ii) integer-only QAT for faithful deployment under FHE constraints, and (iii) secure inference over encrypted behavioral streams, optionally combined with lightweight ensemble aggregation to further improve robustness.

Our main contributions are as follows:
\begin{itemize}
    \item \textbf{Privacy-in-use continuous authentication:} We formulate privacy preserving behavioral continuous authentication under an untrusted inference setting and present an end-to-end FHE-based inference pipeline that keeps sensitive behavioral inputs encrypted throughout computation, addressing privacy concerns highlighted in CA surveys.
    
    \item \textbf{FBAN: FHE-compatible bottleneck attention:} We propose FBAN and an integer-only deployment pipeline via QAT, enabling attention-based modeling to be compiled and executed efficiently under TFHE-style FHE constraints.

    \item \textbf{Strong accuracy–latency trade-offs with rigorous evaluation:} We empirically show that FBAN achieves competitive authentication performance against representative CA baselines (including Transformer-based context-aware models) while reducing model complexity and delivering end-to-end FHE latency practical for periodic or risk-triggered re-verification, supported by an evaluation protocol aligned with realistic threat models.
\end{itemize}

\section{Preliminaries}
\label{sec:preliminaries}

\subsection{Notation}
Scalars are denoted by italic letters (e.g., $x$), vectors by bold lowercase (e.g., $\mathbf{x}$), and matrices by bold uppercase (e.g., $\mathbf{W}$). For an integer $n$, we write $[n]=\{1,\dots,n\}$. We use $\lambda_{\mathsf{sid}}$ for the bit-length of session identifiers and $\ell$ for the bit-length of the one-time pad used in the masked decision return. For a vector $\mathbf{v}$, its $j$-th entry is $\mathbf{v}[j]$. We use $\mathbb{I}[\cdot]$ for the indicator function.

\paragraph{Continuous behavioral authentication stream.}
Let $\mathbf{x}_t\in\mathbb{R}^{d}$ be the feature vector extracted at time step $t$ from behavioral signals (e.g., touch/keystroke/IMU-derived descriptors). A length-$L$ sliding window is
\begin{equation}
\mathbf{X}_t=\big[\mathbf{x}_{t-L+1},\dots,\mathbf{x}_t\big]\in\mathbb{R}^{L\times d}.
\end{equation}
A continuous authenticator with parameters $\theta$ maps each window $\mathbf{X}_t$ to a real-valued score $s_t$ and a binary decision $\hat{y}_t$:
\begin{equation}
s_t=f_{\theta}(\mathbf{X}_t),\qquad \hat{y}_t=\mathbb{I}[s_t\ge \tau],
\end{equation}
where $\tau$ is a decision threshold. Throughout the paper, the score is interpreted as an acceptance confidence (larger is more likely to be genuine), though the construction applies to any monotone scoring rule.

\subsection{Ensemble Learning}
To improve robustness and reduce variance~\cite{sagi2018ensemble}, we consider an ensemble of $N$ base authentication models $\{f_{\theta_i}\}_{i=1}^{N}$. For a given window $\mathbf{X}_t$, each member outputs a score
$s_{t,i}=f_{\theta_i}(\mathbf{X}_t)$, $i\in[N]$.
Let $\mathbf{s}_t=(s_{t,1},\dots,s_{t,N})^\top\in\mathbb{R}^{N}$ collect all member scores. An ensemble aggregator maps $\mathbf{s}_t$ to a final score $\bar{s}_t$, which is then thresholded: $\hat{y}_t=\mathbb{I}[\bar{s}_t\ge \tau]$.
Standard choices include simple averaging and weighted averaging; we adopt \emph{stacking} (a linear meta-model) as it can be directly optimized for an operating constraint.
Let $\phi=(\mathbf{w},c)$ with $\mathbf{w}\in\mathbb{R}^{N}$, $c\in\mathbb{R}$. The stacked score is
\begin{equation}
\bar{s}_t=g_{\phi}(\mathbf{s}_t)=\mathbf{w}^\top \mathbf{s}_t+c.
\end{equation}
The parameters $\phi$ are trained on a held-out set (e.g., minimizing FRR subject to a target FAR).

\subsection{TFHE-Based Fully Homomorphic Encryption}
Fully homomorphic encryption (FHE) enables computation over encrypted data: given ciphertexts encrypting plaintext values, a server can evaluate a function and obtain a ciphertext encrypting the correct result, without learning the plaintexts~\cite{gentry2009fully}. We instantiate FHE with \emph{TFHE}, a gate-bootstrapped FHE family well-suited for boolean and integer arithmetic via programmable bootstrapping~\cite{chillotti2020tfhe}.

\paragraph{TFHE interface.}
Let $\lambda$ be the security parameter. We abstract TFHE as
\begin{equation}
\mathsf{TFHE}=(\mathsf{KeyGen},\mathsf{Enc},\mathsf{Dec},\mathsf{Eval}).
\end{equation}
\begin{itemize}
  \item $\mathsf{TFHE.KeyGen}(1^\lambda)\rightarrow (\mathsf{sk},\mathsf{ek})$ outputs a secret key $\mathsf{sk}$ and an evaluation key $\mathsf{ek}$. The evaluation key contains the auxiliary material required for homomorphic evaluation (e.g., bootstrapping and key-switching keys) and is intended to be shared with the evaluator/server.
  \item $\mathsf{TFHE.Enc}(\mathsf{sk},m)\rightarrow \llbracket m\rrbracket$ encrypts a plaintext message $m$ into a ciphertext $\llbracket m\rrbracket$.
  \item $\mathsf{TFHE.Dec}(\mathsf{sk},\llbracket m\rrbracket)\rightarrow m$ decrypts a ciphertext.
  \item $\mathsf{TFHE.Eval}(\mathsf{ek},\mathcal{C},\cdot)\rightarrow \llbracket \mathcal{C}(\cdot)\rrbracket$ homomorphically evaluates a circuit $\mathcal{C}$ on ciphertext inputs and returns a ciphertext of the output.
\end{itemize}
In our setting, $\mathcal{C}$ represents an \emph{integer-only inference circuit} (including additions, multiplications, comparisons, and piecewise-linear primitives after quantization), which aligns naturally with TFHE's evaluation model.

\paragraph{Correctness and homomorphic evaluation.}
For any supported circuit $\mathcal{C}$ and plaintext $m$, correctness requires
\begin{equation}
\mathsf{TFHE.Dec}\big(\mathsf{sk},\,\mathsf{TFHE.Eval}(\mathsf{ek},\mathcal{C},\llbracket m\rrbracket)\big)=\mathcal{C}(m),
\end{equation}
except with negligible probability in $\lambda$.

For clarity, we also write homomorphic evaluation for basic integer operations using circuits $\mathcal{C}_{+}$ and $\mathcal{C}_{\times}$:
\begin{align}
\llbracket m_1+m_2\rrbracket &\leftarrow \mathsf{TFHE.Eval}(\mathsf{ek},\mathcal{C}_{+},\llbracket m_1\rrbracket,\llbracket m_2\rrbracket), \label{eq:tfhe-add}\\
\llbracket m_1\cdot m_2\rrbracket &\leftarrow \mathsf{TFHE.Eval}(\mathsf{ek},\mathcal{C}_{\times},\llbracket m_1\rrbracket,\llbracket m_2\rrbracket). \label{eq:tfhe-mul}
\end{align}
% These equations emphasize the operational view used later: the client encrypts inputs once, and the server evaluates a fixed inference circuit on ciphertexts to produce an encrypted score.

\subsection{Integer Quantization for TFHE-Compatible Neural Inference}
TFHE evaluation is most practical when the target computation is expressed over integers or bits. We therefore adopt a standard uniform affine quantization to convert real-valued tensors into $b$-bit integers, enabling an \emph{integer-only network} whose forward pass can be compiled into a circuit $\mathcal{C}_\theta$.

Let $\Delta>0$ be a scale and $z\in\mathbb{Z}$ be a zero-point. For a real scalar $x$, define the $b$-bit quantizer
\begin{equation}
Q_b(x)=\mathrm{clip}\Big(\big\lfloor x/\Delta \big\rceil+z,\; 0,\;2^b-1\Big)\in\mathbb{Z}_{2^b}, \label{eq:quant}
\end{equation}
where $\lfloor\cdot\rceil$ denotes rounding to the nearest integer and $\mathrm{clip}(\cdot)$ clamps values into the representable range. The corresponding de-quantization is
\begin{equation}
\tilde{x}=D_b(q)=\Delta\cdot(q-z),\qquad q\in\mathbb{Z}_{2^b}. \label{eq:dequant}
\end{equation}
For vectors/matrices, $Q_b(\cdot)$ and $D_b(\cdot)$ are applied element-wise.

We denote the quantized (integer-only) version of a network by $f^{(Q)}_{\theta}$. Its forward computation is represented as an integer circuit $\mathcal{C}_{\theta}$ over $\mathbb{Z}_{2^b}$. Given an input window $\mathbf{X}_t$, we quantize it to $\mathbf{Q}_t=Q_b(\mathbf{X}_t)$ and compute an integer score
\begin{equation}
s^{(Q)}_t=\mathcal{C}_{\theta}(\mathbf{Q}_t)\in\mathbb{Z}_{2^b}. \label{eq:int-score}
\end{equation}
Encrypted inference under TFHE then evaluates the same circuit on ciphertext inputs:
\begin{equation}
\llbracket s^{(Q)}_t\rrbracket
=\mathsf{TFHE.Eval}\big(\mathsf{ek},\mathcal{C}_{\theta},\llbracket \mathbf{Q}_t\rrbracket\big). \label{eq:fhe-score}
\end{equation}
This notation will be reused when specifying the end-to-end protocol: the client produces $\llbracket \mathbf{Q}_t\rrbracket$, the server evaluates $\mathcal{C}_\theta$ homomorphically, and the client decrypts to obtain $s_t^{(Q)}$ (or a post-processed decision) while the server learns no plaintext behavioral features.

\section{Framework Design}
\label{sec:framework}

\noindent
This section details the proposed \textbf{FBAN} framework from both the \emph{model-design} and the \emph{cryptographic deployment} perspectives.
We first formalize the client--server workflow and the per-user template/policy representation, then present the TFHE-compatible FBAN architecture,
and finally specify the end-to-end interaction protocol that supports privacy-preserving continuous authentication.
In particular, the complete \emph{offline} procedure that produces a user-specific template---including two-stage training (pre-training $\rightarrow$ QAT),
TFHE compilation, and fixed-point stacking calibration---is summarized in Algorithm~\ref{alg:template}.
Given the resulting template, the \emph{online} blinded authentication interaction executed at each time step is specified in Algorithm~\ref{alg:protocol}.

\subsection{System Overview}
\label{subsec:framework-overview}
Fig.~\ref{fig:fban_overview} depicts the end-to-end workflow of \textbf{FBAN}.
At each time step $t$, the client device $\mathcal{C}$ extracts a behavioral window
$\mathbf{X}_t\in\mathbb{R}^{L\times d}$ and applies a deterministic preprocessing pipeline
$\mathsf{Norm}(\cdot)$ (e.g., cleaning and standardization). We treat $\mathsf{Norm}$ as a public, deterministic transformation with fixed parameters (e.g., means/standard deviations) learned offline and included in public parameters. This avoids any ambiguity about hidden state in preprocessing. Followed by $b$-bit quantization $Q_b(\cdot)$,
to obtain an integer tensor
\begin{equation}
\tilde{\mathbf{X}}_t := Q_b\!\big(\mathsf{Norm}(\mathbf{X}_t)\big)\in \mathbb{Z}_{2^b}^{L\times d}.
\label{eq:fw-quant-window}
\end{equation}
The client encrypts $\tilde{\mathbf{X}}_t$ under TFHE and sends only ciphertexts to the server $\mathcal{S}$ for inference.
The server evaluates TFHE-compiled FBAN circuits and returns an encrypted one-time-padded decision bit for decrypt-and-return verification
(Algorithm~\ref{alg:protocol}).

\paragraph{Template/policy.}
For each user identity $\mathsf{uid}$, the server stores a template/policy
\[
\mathsf{T}_{\mathsf{uid}} := \Big(\{\mathcal{C}_{\mathsf{uid},i}\}_{i=1}^{N},\ \hat{\mathbf{w}}_{\mathsf{uid}},\ \hat{\tau}_{\mathsf{uid}}\Big),
\]
where $\{\mathcal{C}_{\mathsf{uid},i}\}$ are TFHE inference circuits of an ensemble of size $N$,
$\hat{\mathbf{w}}_{\mathsf{uid}}=(\hat{w}_{\mathsf{uid},1},\ldots,\hat{w}_{\mathsf{uid},N})\in\mathbb{Z}^N$ are fixed-point stacking weights,
and $\hat{\tau}_{\mathsf{uid}}\in\mathbb{Z}$ is the correspondingly scaled decision threshold chosen at a target operating point (e.g., FAR $=\rho$).
The template is generated offline following Algorithm~\ref{alg:template} and then reused for all subsequent online authentications.

\paragraph{Design goals.}
FBAN targets:
(i) \textbf{privacy-in-use} for behavioral windows during inference via TFHE;
(ii) \textbf{FHE-compatibility} via integer-only computation and circuit-friendly attention;
(iii) \textbf{practicality} via high authentication accuracy and encrypted latency suitable for periodic or risk-triggered re-verification.

\begin{figure}[htbp]
    \centering
    \includegraphics[width=\linewidth]{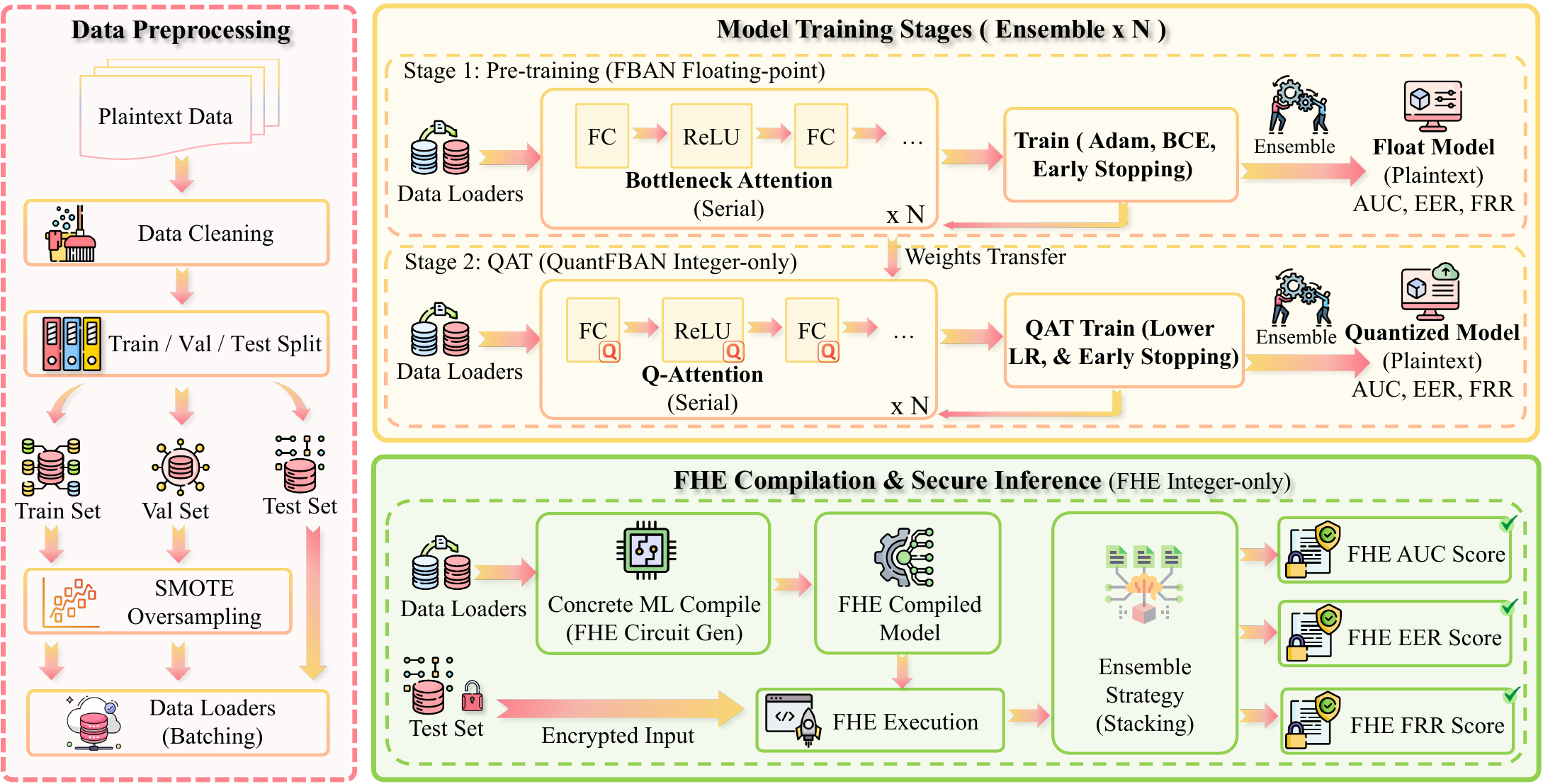}
    \caption{FBAN workflow: preprocessing, two-stage training (pre-training $\rightarrow$ QAT), TFHE compilation, and secure inference with ensemble stacking.}
    \label{fig:fban_overview}
\end{figure}

\subsection{TFHE-Compatible Bottleneck Attention Network}
\label{subsec:framework-fban}
Let $\tilde{\mathbf{x}}_t=\mathrm{vec}(\tilde{\mathbf{X}}_t)\in\mathbb{Z}_{2^b}^{D}$ with $D=L\cdot d$.
A naive attention layer (softmax normalization and multiplicative gating) is ill-suited for TFHE due to high circuit cost.
FBAN therefore employs \emph{serial bottleneck attention} that uses only affine maps, an integer-friendly activation, and additions,
which can be compiled into efficient TFHE circuits.

\paragraph{Bottleneck attention block.}
Given an intermediate hidden vector $\mathbf{h}\in\mathbb{Z}_{2^b}^{m}$,
a bottleneck attention block is defined as
\begin{equation}
\label{eq:fw-ba}
\mathsf{BA}(\mathbf{h})
~:=~
\act\!\Big(\mathbf{W}_{e}\,\act(\mathbf{W}_{s}\mathbf{h}+\mathbf{b}_{s})+\mathbf{b}_{e}\Big),
\end{equation}
where $\mathbf{W}_{s}\in\mathbb{Z}_{2^b}^{m'\times m}$, $\mathbf{W}_{e}\in\mathbb{Z}_{2^b}^{m\times m'}$ with $m'\ll m$,
and $\act(\cdot)$ denotes an integer-compatible activation (e.g., quantized ReLU) realizable under TFHE.
Importantly, \eqref{eq:fw-ba} avoids ciphertext--ciphertext multiplicative gating, keeping TFHE evaluation efficient.

\paragraph{FBAN forward computation.}
A representative instance (matching our implementation) is:
\begin{align}
\label{eq:fw-fwd}
\mathbf{h}_1 &:= \act(\mathbf{W}_1 \tilde{\mathbf{x}}_t + \mathbf{b}_1), \nonumber\\
\mathbf{h}_2 &:= \act(\mathbf{W}_2 \mathbf{h}_1 + \mathbf{b}_2), \nonumber\\
\mathbf{h}_3 &:= \mathsf{BA}(\mathbf{h}_2), \nonumber\\
\mathbf{h}_4 &:= \act(\mathbf{W}_3 \mathbf{h}_3 + \mathbf{b}_3), \nonumber\\
\mathbf{h}_5 &:= \act(\mathbf{W}_4 \mathbf{h}_4 + \mathbf{b}_4), \nonumber\\
\mathbf{h}_6 &:= \mathsf{BA}(\mathbf{h}_5), \nonumber\\
s_t &:= \mathbf{w}_{o}^{\top}\act(\mathbf{W}_5 \mathbf{h}_6 + \mathbf{b}_5)+b_o \in \mathbb{Z},
\end{align}
where $s_t$ is the integer logit score. All parameters are trained under quantization constraints so that every operation in
\eqref{eq:fw-fwd} is valid over bounded integers, enabling direct compilation to TFHE circuits (Algorithm~\ref{alg:template}).

\subsection{Offline Model Preparation for TFHE Inference}
\label{subsec:framework-train}
This subsection describes how a per-user template $\mathsf{T}_{\mathsf{uid}}$ is constructed offline.
Given the enrollment dataset $\mathcal{D}_{\mathsf{uid}}=\{(\mathbf{X},y)\}$ with $y\in\{0,1\}$,
FBAN is trained in two stages: (i) floating-point pre-training to obtain a strong initialization, and
(ii) quantization-aware training (QAT) to produce an integer-only network.
After QAT, each member is compiled into a TFHE circuit using a calibration subset (for range/quantization calibration).

\paragraph{FHE-friendly stacking.}
Let $s^{(i)}_t$ be the plaintext score of member $i$ at time $t$ (Eq.~\eqref{eq:fw-fwd}).
We aggregate members using a stacking-style linear combiner
\begin{equation}
\label{eq:fw-stack}
S_t ~:=~ \sum_{i=1}^{N} \hat{w}_{\mathsf{uid},i}\cdot s^{(i)}_t,
\end{equation}
where $\hat{w}_{\mathsf{uid},i}\in\mathbb{Z}$ are fixed-point weights learned on the validation set to optimize the operating metric (e.g., FRR@FAR).
The decision threshold $\hat{\tau}_{\mathsf{uid}}$ is chosen under the same fixed-point convention.
The complete template generation pipeline---including data preparation, two-stage training, TFHE compilation, and stacking calibration---is
summarized in Algorithm~\ref{alg:template}.

\begin{algorithm}[t]
\caption{Template Generation for User $\mathsf{uid}$ (Two-Stage Training $\rightarrow$ TFHE Circuits $\rightarrow$ Stacking Policy)}
\label{alg:template}
\KwIn{Enrollment dataset $\mathcal{D}_{\mathsf{uid}}$; ensemble size $N$; quantization bit-width $b$; target FAR $\rho$}
\KwOut{Template $\mathsf{T}_{\mathsf{uid}}=(\{\mathcal{C}_{\mathsf{uid},i}\}_{i=1}^{N},\hat{\mathbf{w}}_{\mathsf{uid}},\hat{\tau}_{\mathsf{uid}})$}

\BlankLine
\AlgPhase{Data Preparation}
Split $\mathcal{D}_{\mathsf{uid}}$ into $(\mathcal{D}^{\mathrm{tr}},\mathcal{D}^{\mathrm{val}},\mathcal{D}^{\mathrm{te}})$\;
(Optional) Apply class re-balancing on $\mathcal{D}^{\mathrm{tr}}$ (e.g., SMOTE)\AlgCmt{Train-only to avoid contamination}\;

\BlankLine
\For{$i\leftarrow 1$ \KwTo $N$}{
\AlgPhase{Stage 1: Pre-training (Floating-Point FBAN)}
Train FBAN$_i$ on $\mathcal{D}^{\mathrm{tr}}$ with early stopping on $\mathcal{D}^{\mathrm{val}}$\AlgCmt{e.g., Adam + BCE}\;

\AlgPhase{Stage 2: QAT (Integer-Only QuantFBAN)}
Initialize QuantFBAN$_i$ by weight transfer from FBAN$_i$\AlgCmt{Warm-start improves convergence}\;
QAT-train QuantFBAN$_i$ with $b$-bit quantization on $\mathcal{D}^{\mathrm{tr}}$\;

\AlgPhase{TFHE Compilation}
Select a calibration subset $\mathcal{D}^{\mathrm{cal}}\subseteq \mathcal{D}^{\mathrm{tr}}$\AlgCmt{Range calibration}\;
Compile QuantFBAN$_i$ into a TFHE circuit $\mathcal{C}_{\mathsf{uid},i}$\AlgCmt{Circuit gen + parameterization}\;
}

\BlankLine
\AlgPhase{Stacking Policy (Fixed-Point)}
Compute validation scores $\{s^{(i)}_t\}$ from members in plaintext simulation\;
Learn $\hat{\mathbf{w}}_{\mathsf{uid}}\in\mathbb{Z}^{N}$ to minimize FRR@FAR$(\rho)$ under \eqref{eq:fw-stack}\;
Select threshold $\hat{\tau}_{\mathsf{uid}}$ on $\mathcal{D}^{\mathrm{val}}$ at FAR $=\rho$\;

\Return $\mathsf{T}_{\mathsf{uid}}=(\{\mathcal{C}_{\mathsf{uid},i}\}_{i=1}^{N},\hat{\mathbf{w}}_{\mathsf{uid}},\hat{\tau}_{\mathsf{uid}})$\;
\end{algorithm}

\subsection{Blinded TFHE Continuous Authentication Protocol}
\label{subsec:framework-protocol}
This subsection specifies the online interaction between $\mathcal{C}$ and $\mathcal{S}$ for continuous authentication.
The client holds the TFHE secret key $\mathsf{sk}$ and decrypts, while the server holds the evaluation key $\mathsf{ek}$ and the user template.
At each time step $t$, the server enforces \emph{fresh blinding} to (i) bind submissions to the current session and
(ii) prevent a malicious client from deterministically forging an always-accept response based on raw scores.

\paragraph{Input blinding (session binding).}
The server samples a fresh additive mask $\boldsymbol{\beta}_t\in\mathbb{Z}_{2^b}^{L\times d}$ and requires the client to encrypt
$\tilde{\mathbf{X}}_t^\star := \tilde{\mathbf{X}}_t + \boldsymbol{\beta}_t$.
The server removes $\boldsymbol{\beta}_t$ \emph{inside TFHE} by evaluating a constant-subtraction circuit.

\paragraph{Masked decision return (anti-forgery and minimal disclosure).}
To avoid leaking the raw score $S_t$ and to prevent a malicious client from \emph{forging} the server-side accept/reject outcome, the server performs the threshold comparison under TFHE and returns only a one-time padded \emph{decision bit}.
Let
\[
d_t := \mathbb{I}\!\left[S_t \ge \hat{\tau}_{\mathsf{uid}}\right]\in\{0,1\}.
\]
The server samples a fresh pad $r_t \xleftarrow{\$} \mathbb{Z}_{2^\ell}$ (kept secret from the client) and returns a ciphertext of
\begin{equation}
\label{eq:fw-blindout}
U_t := r_t + d_t \pmod{2^\ell}.
\end{equation}
We encode the bit $d_t$ as an integer in $\mathbb{Z}_{2^\ell}$ (i.e., $0$ or $1$) so that the addition in Eq.~\eqref{eq:fw-blindout} is well-defined. The client decrypts $U_t$ and sends it back in the clear. The server recovers $d_t$ by subtracting $r_t$ (mod $2^\ell$) and outputs \textsf{Accept} iff $d_t=1$.
This design yields a standard \emph{decrypt-and-return} soundness check while revealing neither $S_t$ nor $d_t$ to the client device.

The complete online protocol is detailed in Algorithm~\ref{alg:protocol}.

\begin{algorithm}[htbp]
\SetAlFnt{\scriptsize}
\SetAlgoNlRelativeSize{-2}
\setlength{\algomargin}{0.6em}
\SetInd{0.25em}{0.6em}
\renewcommand{\baselinestretch}{0.93}\selectfont
\caption{Protocol $\Pi_{\mathrm{FBAN}}$: TFHE-Based Privacy-Preserving Continuous Authentication (Blinded)}
\label{alg:protocol}
\KwIn{User identity $\mathsf{uid}$; template $\mathsf{T}_{\mathsf{uid}}=(\{\mathcal{C}_{\mathsf{uid},i}\}_{i=1}^{N},\hat{\mathbf{w}}_{\mathsf{uid}},\hat{\tau}_{\mathsf{uid}})$}
\KwOut{Server decision \textsf{Accept} / \textsf{Reject} at each time step $t$}

\BlankLine
\tcp*[l]{Key setup and enrollment are performed as in Alg.~\ref{alg:template}; we describe the online phase.}
\BlankLine
\ForEach{time step $t$}{
\Server{
\AlgPhase{Fresh Blinding Material}
Sample fresh session id $\mathsf{sid}_t \xleftarrow{\$} \{0,1\}^{\lambda_{\mathsf{sid}}}$ and record it as used\;
$\boldsymbol{\beta}_t \xleftarrow{\$} \mathbb{Z}_{2^b}^{L\times d}$\AlgCmt{Input one-time mask}\;
$r_t \xleftarrow{\$} \mathbb{Z}_{2^\ell}$\AlgCmt{Secret one-time pad for masked decision return}\;
Send $(\mathsf{sid}_t,\boldsymbol{\beta}_t)$ to $\mathcal{C}$\AlgCmt{$r_t$ kept secret}\;
}

\Client{
\AlgPhase{Quantize, Blind, Encrypt}
Extract window $\mathbf{X}_t$ and compute $\tilde{\mathbf{X}}_t := Q_b(\mathsf{Norm}(\mathbf{X}_t))$\AlgCmt{Preprocessing}\;
$\tilde{\mathbf{X}}_t^\star \leftarrow \tilde{\mathbf{X}}_t + \boldsymbol{\beta}_t$\AlgCmt{Input blinding}\;
$\llbracket \tilde{\mathbf{X}}_t^\star \rrbracket \leftarrow \mathsf{TFHE.Enc}(\mathsf{sk},\tilde{\mathbf{X}}_t^\star)$\;
Send $(\mathsf{uid},\mathsf{sid}_t,\llbracket \tilde{\mathbf{X}}_t^\star \rrbracket)$ to $\mathcal{S}$\;
}

\Server{
\AlgPhase{Session Check}
\If{$\mathsf{sid}_t$ is not fresh}{Output \textsf{Reject}\;}
\AlgPhase{In-Circuit Unblinding}
\KwLet{$\mathcal{C}_{-\boldsymbol{\beta}_t}$}{the circuit $\mathbf{U}\mapsto \mathbf{U}-\boldsymbol{\beta}_t$ over $\mathbb{Z}_{2^b}^{L\times d}$}\AlgCmt{Constant subtraction}\;
$\llbracket \tilde{\mathbf{X}}_t \rrbracket \leftarrow \mathsf{TFHE.Eval}(\mathsf{ek},\mathcal{C}_{-\boldsymbol{\beta}_t},\llbracket \tilde{\mathbf{X}}_t^\star \rrbracket)$\;

\AlgPhase{TFHE Inference (Ensemble)}
\For{$i\leftarrow 1$ \KwTo $N$}{
$\llbracket s_t^{(i)} \rrbracket \leftarrow \mathsf{TFHE.Eval}(\mathsf{ek},\mathcal{C}_{\mathsf{uid},i},\llbracket \tilde{\mathbf{X}}_t \rrbracket)$\AlgCmt{FBAN member $i$}\;
}

\AlgPhase{Stacking in TFHE}
$\llbracket S_t \rrbracket \leftarrow \sum_{i=1}^{N}\hat{w}_{\mathsf{uid},i}\cdot \llbracket s_t^{(i)} \rrbracket$\AlgCmt{Eq.~\eqref{eq:fw-stack}}\;

\AlgPhase{Masked Decision in TFHE}
\KwLet{$\mathcal{C}_{\ge \hat{\tau}_{\mathsf{uid}}}$}{the circuit $z\mapsto \mathbb{I}[z\ge \hat{\tau}_{\mathsf{uid}}]$ outputting a bit $d_t\in\{0,1\}$}\;
$\llbracket d_t \rrbracket \leftarrow \mathsf{TFHE.Eval}(\mathsf{ek},\mathcal{C}_{\ge \hat{\tau}_{\mathsf{uid}}},\llbracket S_t \rrbracket)$\;
\KwLet{$\mathcal{C}_{+r_t}$}{the circuit $u\mapsto u+r_t$ over $\mathbb{Z}_{2^\ell}$}\AlgCmt{Eq.~\eqref{eq:fw-blindout}}\;
$\llbracket U_t \rrbracket \leftarrow \mathsf{TFHE.Eval}(\mathsf{ek},\mathcal{C}_{+r_t},\llbracket d_t \rrbracket)$\;
Send $(\mathsf{sid}_t,\llbracket U_t \rrbracket)$ to $\mathcal{C}$\;
}

\Client{
\AlgPhase{Decrypt-and-Return}
$U_t \leftarrow \mathsf{TFHE.Dec}(\mathsf{sk},\llbracket U_t \rrbracket)$\AlgCmt{Masked decision; reveals neither $S_t$ nor $d_t$}\;
Send $(\mathsf{sid}_t,U_t)$ to $\mathcal{S}$\;
}

\Server{
\AlgPhase{Verification and Decision}
\If{$\mathsf{sid}_t$ was seen before for this $t$}{Output \textsf{Reject}\;}\AlgCmt{Replay protection}
$d'_t \leftarrow (U_t - r_t)\bmod 2^\ell$\;
\If{$d'_t = 1$}{Output \textsf{Accept}\;}\Else{Output \textsf{Reject}\;}
}
}
\end{algorithm}

%======================================================================
% Threat Model (Honest-but-Curious Server)
% - Cryptography-style, symbol-heavy, reuses existing notation:
%   C,S, uid, X_t, \tilde{X}_t, Q_b, Norm, ek, sk, T_uid, sid_t,
%   beta_t, alpha_t, gamma_t, r_t, s_t^{(i)}, S_t, Z_t, \llbracket \cdot \rrbracket
% - Designed to plug after Framework Design and before Experiments
%======================================================================

\section{Threat Model and Security Goals}
\label{sec:threat}

\paragraph{System model and trust assumptions.}
We consider the two-party setting induced by Protocol~$\Pi_{\mathrm{FBAN}}$ (Algorithm~\ref{alg:protocol}),
between a client device $\mathcal{C}$ (holding the behavioral stream) and an authentication server $\mathcal{S}$.
At each time step $t$, $\mathcal{C}$ extracts a behavioral window $\mathbf{X}_t\in\mathbb{R}^{L\times d}$ and deterministically computes
the quantized integer representation $\tilde{\mathbf{X}}_t := Q_b(\mathsf{Norm}(\mathbf{X}_t))\in\mathbb{Z}_{2^b}^{L\times d}$
(see Eq.~\eqref{eq:fw-quant-window}).
$\mathcal{C}$ runs TFHE key generation $(\mathsf{sk},\mathsf{ek})\leftarrow \mathsf{TFHE.KeyGen}(1^\lambda)$,
keeps $\mathsf{sk}$ locally, and registers $\mathsf{ek}$ to $\mathcal{S}$.
For each identity $\mathsf{uid}$, $\mathcal{S}$ stores the user template/policy
$\mathsf{T}_{\mathsf{uid}}=(\{\mathcal{C}_{\mathsf{uid},i}\}_{i=1}^{N},\hat{\mathbf{w}}_{\mathsf{uid}},\hat{\tau}_{\mathsf{uid}})$
generated offline (Algorithm~\ref{alg:template}).
We assume authenticated communication (e.g., standard channel authentication) to prevent active man-in-the-middle modification.
Confidentiality of the channel is \emph{not} required for privacy, since all sensitive client inputs are encrypted under TFHE.

\paragraph{Adversary model (honest-but-curious server).}
We model $\mathcal{S}$ as an \emph{honest-but-curious} (a.k.a.\ semi-honest) adversary:
it follows Algorithm~\ref{alg:protocol} exactly, but attempts to infer sensitive information about $\{\mathbf{X}_t\}$ (and hence $\{\tilde{\mathbf{X}}_t\}$)
from its internal state and the protocol transcript.
Formally, let $\mathcal{A}_{\mathcal{S}}$ be a PPT algorithm controlling the server.
For each time step $t$, define the server-side transcript item
\begin{equation}
\label{eq:th-tr}
\mathsf{tr}_t :=
\Big(\mathsf{uid},\mathsf{sid}_t,\boldsymbol{\beta}_t,\llbracket \tilde{\mathbf{X}}_t^\star\rrbracket,\llbracket U_t\rrbracket\Big),
\end{equation}
where $\tilde{\mathbf{X}}_t^\star=\tilde{\mathbf{X}}_t+\boldsymbol{\beta}_t$ is the client-side blinded input (Algorithm~\ref{alg:protocol}),
and $\llbracket U_t\rrbracket$ is the ciphertext of the masked decision value from Eq.~\eqref{eq:fw-blindout}.
Since $\mathcal{A}_{\mathcal{S}}$ controls $\mathcal{S}$, its \emph{view} additionally contains all server-held secrets and randomness, namely
\begin{equation}
\label{eq:th-view}
\mathsf{View}_{\mathcal{S}} :=
\Big(\mathsf{ek}, \mathsf{T}_{\mathsf{uid}}, \{r_t\}_t,\ \{\mathsf{tr}_t\}_t \Big),
\end{equation}
as well as any intermediate ciphertexts deterministically derived during TFHE evaluation (e.g., $\llbracket s_t^{(i)}\rrbracket$ and $\llbracket S_t\rrbracket$),
which are functions of $\mathsf{ek}$, $\mathsf{T}_{\mathsf{uid}}$, and $\llbracket \tilde{\mathbf{X}}_t^\star\rrbracket$.

\paragraph{Security goal (privacy-in-use of behavioral windows).}
Our primary goal is to ensure \emph{input privacy} against $\mathcal{A}_{\mathcal{S}}$:
the server should learn nothing about the client windows beyond inevitable leakage such as public dimensions $(L,d)$, bit-width $b$,
the number of time steps, and the final authentication decision.
We capture this goal via an indistinguishability-based definition on the server view.

Let $\mathbf{X}_{1:T}^{(0)}$ and $\mathbf{X}_{1:T}^{(1)}$ be two sequences of windows of equal length $T$ and identical public meta-data
(i.e., same $(L,d)$, same quantization bit-width $b$, and the same protocol parameters).
Let $\mathsf{View}_{\mathcal{S}}^{(0)}$ (resp.\ $\mathsf{View}_{\mathcal{S}}^{(1)}$) denote the distribution of the server view
(Eq.~\eqref{eq:th-view}) when Protocol~$\Pi_{\mathrm{FBAN}}$ (Algorithm~\ref{alg:protocol}) is executed on $\mathbf{X}_{1:T}^{(0)}$
(resp.\ $\mathbf{X}_{1:T}^{(1)}$).
We say $\Pi_{\mathrm{FBAN}}$ provides \emph{server-side input privacy} if for every PPT distinguisher $\mathcal{D}$,
\begin{equation}
\label{eq:th-ind}
\Big|
\Pr\big[\mathcal{D}(\mathsf{View}_{\mathcal{S}}^{(0)})=1\big]
-
\Pr\big[\mathcal{D}(\mathsf{View}_{\mathcal{S}}^{(1)})=1\big]
\Big|
\le \mathsf{negl}(\lambda).
\end{equation}
Intuitively, even though $\mathcal{S}$ observes the ciphertext stream and performs TFHE evaluation, it cannot distinguish which underlying behavioral
sequence generated the transcript, except with negligible advantage in the TFHE security parameter $\lambda$.

\paragraph{Leakage function and scope.}
For clarity, we make explicit the permitted leakage.
Let $\mathcal{L}$ be the leakage function that returns only public meta-data and the server output decisions:
\begin{equation}
\label{eq:th-leak}
\mathcal{L}(\mathbf{X}_{1:T}) :=
\Big(\mathsf{uid},\, T,\, L,\, d,\, b,\, N,\, \{ \mathsf{sid}_t \}_t,\, \{ \mathsf{Dec}_t \}_t \Big),
\end{equation}
where $\mathsf{Dec}_t\in\{\textsf{Accept},\textsf{Reject}\}$ is the server decision at time $t$.
The guarantee in Eq.~\eqref{eq:th-ind} is understood \emph{modulo} $\mathcal{L}$:
the protocol does not aim to hide $(L,d,b,T)$ nor the final decision sequence, since these are either public configuration values or explicit outputs.
We also emphasize that side channels (e.g., fine-grained timing, cache attacks, or physical compromise of $\mathcal{C}$ revealing $\mathsf{sk}$)
are outside the scope of this work.

\paragraph{Remark on blinding under the semi-honest server model.}
The blinding values $(\boldsymbol{\beta}_t,r_t)$ are generated by $\mathcal{S}$ and thus are included in $\mathsf{View}_{\mathcal{S}}$.
Accordingly, under the honest-but-curious server model, blinding is \emph{not} introduced to further hide inputs from $\mathcal{S}$,
but rather to (i) avoid unnecessary exposure of raw scores to the client and (ii) support a unified protocol interface that can be extended to
stronger adversarial models without changing the TFHE inference core.
The server-side privacy goal of this paper remains Eq.~\eqref{eq:th-ind}, grounded in the semantic security of TFHE encryption.

\section{Experimental Evaluation}
\label{sec:exp}

This section evaluates \textbf{FBAN} in terms of both \emph{authentication utility} and \emph{cryptographic deployability}. 
% We (i) compare FBAN with representative methods in touchscreen-based continuous authentication, reporting AUC, EER, and FRR@FAR$=0.001$ together with Params and FLOPs, and (ii) analyze the effect of the ensemble size $N$ by reporting \textsc{Pretrain}, \textsc{QAT}, and \textsc{TFHE} results across varying $N$.
We (i) compare FBAN with representative methods in touchscreen-based continuous authentication, reporting AUC (area under the receiver operating characteristic curve), EER (equal error rate), and FRR@FAR$=0.001$ (false reject rate at false accept rate $0.001$) together with Params and FLOPs, and (ii) analyze the effect of the ensemble size $N$ by reporting \textsc{Pretrain}, \textsc{QAT}, and \textsc{TFHE} results across varying $N$.

\textbf{Datasets.}
We use two public touchscreen-based continuous authentication datasets and follow a unified split protocol to ensure fair and reproducible evaluation.
\emph{TouchPitfalls}~\cite{AsiaCCS2022} is collected on iPhone devices from 470 users in a remote setting, where participants perform two realistic scenarios: social-media style vertical scrolling and gallery-like horizontal swiping.
\emph{Touchalytics}~\cite{Touchalytics} is collected on Android devices from 41 users, covering both natural browsing/reading interactions and a controlled image-comparison task that elicits repeated swipe gestures.
For each dataset, we partition the samples into disjoint subsets with a $65\%/15\%/20\%$ ratio for training/validation/test, respectively; the validation split is used for model selection and threshold calibration, while all reported results are obtained on the held-out test split.

% \textbf{Datasets.}
% We use two public touchscreen-based continuous authentication datasets.
% \emph{TouchPitfalls}~\cite{AsiaCCS2022} is collected on iPhone devices from 470 users in a remote setting, where participants perform two realistic scenarios: social-media style vertical scrolling and gallery-like horizontal swiping.
% \emph{Touchalytics}~\cite{Touchalytics} is collected on Android devices from 41 users, covering both natural browsing/reading interactions and a controlled image-comparison task that elicits repeated swipe gestures.

% \paragraph{Baselines.}
% We compare FBAN against a broad set of baselines spanning (i) \emph{classical} methods such as SVM-based models with hand-crafted behavioral features, (ii) \emph{CNN-based} time-series encoders, and (iii) \emph{Transformer-based} backbones for sequential modeling.

\textbf{Metrics.}
We report the area under the ROC curve (AUC), equal error rate (EER), and false reject rate (FRR) at a fixed operating point (i.e., FRR@FAR$=\rho$).
To quantify deployability, we additionally report the number of parameters (Params) and computational cost (FLOPs/PLOPs, following a consistent counting convention across methods).

\textbf{Implementation.}
% All experiments are run on an Apple M3 Max processor with 16-core CPU, 48-core GPU, up to 4.05\,GHz and 48\,GB unified memory. Implementations were in Python~3.9.21 using PyTorch~1.13.1 for deep learning and \texttt{Concrete-ML}~\cite{ConcreteML} for TFHE-based homomorphic encryption~\cite{chillotti2020tfhe}.

All experiments are run on a Linux server equipped with an Intel Xeon Platinum 8352S CPU (32 cores, 64 threads) and an NVIDIA L40 GPU with 48 GB of memory. Implementations were in Python~3.9.21 using PyTorch~1.13.1 for deep learning and \texttt{Concrete-ML}~\cite{ConcreteML} for TFHE-based homomorphic encryption~\cite{chillotti2020tfhe}.

\textbf{Experimental configuration.}
For both datasets, each input sample is constructed by concatenating $L=3$ consecutive strokes into a single behavioral window, yielding flattened input dimensions $D = L \times d = 3 \times 27 = 81$ for \textsc{TouchPitfalls} and $D = 3 \times 41 = 123$ for \textsc{Touchalytics}, where $d$ denotes the number of per-stroke features covering stroke geometry/timing, kinematics, touch-signal statistics, and device orientation.
For each target user, impostor (negative) samples are constructed by pooling the strokes of \emph{all remaining users} in the same dataset split, preserving the natural impostor distribution without cross-split leakage.
The decision threshold $\hat{\tau}_{\mathsf{uid}}$ is calibrated per user on the validation split at the target operating point: we compute the ROC curve over validation scores, use linear interpolation to identify the threshold at which FAR equals the target ($\rho=0.001$), and fix this threshold when reporting FRR on the held-out test split.

\subsection{Comparative Analysis}
\label{subsec:exp-baselines}

We compare \textbf{FBAN} with representative baselines including classical SVM-based methods~\cite{AsiaCCS2022} ~\cite{SVMDA} and recent deep backbones (CNN-~\cite{OSCNN} and Transformer-based~\cite{TSCA} models). We report AUC, EER, and FRR at FAR$=0.001$ on both \textsc{TouchPitfalls} and \textsc{Touchalytics}, and summarize model complexity by FLOPs and Params when available.

% =============================== (Main table) ================================
\begin{table}[htbp]
\centering
\caption{Overall comparison on \textsc{TouchPitfalls} and \textsc{Touchalytics} datasets.
The \best{best} and \second{second-best} results are highlighted in color.}
\label{tab:overall-one}

% --- make table compact ---
\scriptsize
\setlength{\tabcolsep}{2.6pt}
\renewcommand{\arraystretch}{1.05}

\begin{tabular}{@{}l ccc ccc cc@{}}
\toprule
\multirow{2}{*}{\textbf{Method}}
& \multicolumn{3}{c}{\textbf{TouchPitfalls~\cite{AsiaCCS2022}}}
& \multicolumn{3}{c}{\textbf{Touchalytics~\cite{Touchalytics}}}
& \multirow{2}{*}{\textbf{FLOPs}}
& \multirow{2}{*}{\textbf{Params}} \\
\cmidrule(lr){2-4}\cmidrule(lr){5-7}
& \textbf{AUC}$\uparrow$ & \textbf{FRR}$^{*}\downarrow$ & \textbf{EER}$\downarrow$
& \textbf{AUC}$\uparrow$ & \textbf{FRR}$^{*}\downarrow$ & \textbf{EER}$\downarrow$
& & \\
\midrule
SVM\_TA$^{\dagger}$~\cite{AsiaCCS2022}  & 93.97 & 18.47 & 8.66 & 98.81 & 10.34 & 4.22 & -- & -- \\
SVM\_DA$^{\dagger}$~\cite{SVMDA}  & 92.80 & 31.43 & 9.49 & 98.79 & 13.07 & 3.58 & -- & -- \\
PatchTST$^{\dagger}$~\cite{PatchTST} & 95.76 & 13.82 & 5.95 & 97.82 & 6.31  & 5.57 & \second{0.078G} & 0.313M \\
TimesNet$^{\dagger}$~\cite{Timesnet} & 96.33 & 10.79 & 4.51 & 99.11 & 5.48  & \second{1.55} & 13.709G & 8.288M \\
OS\_CNN$^{\dagger}$~\cite{OSCNN} & 95.85 & 12.17 & 5.16 & 99.14 & 6.52  & 2.28 & 0.382G & 0.763M \\
TSCA$^{\dagger}$~\cite{TSCA}     & \second{97.59} & \best{8.16} & \second{2.42} & \second{99.45} & \best{3.62} & \best{1.12} & {0.098G} & \second{0.225M} \\
\midrule
FBAN (Ours)$^{\ddagger}$ & \best{98.93} & \second{9.07} & \best{1.18} & \best{99.87} & \second{5.07} & 1.65 & \best{$3.44\times10^{-3}$G} & \best{0.0271M} \\
\bottomrule
\end{tabular}
% \vspace{0.15em}
{\tiny
\begin{flushleft}
$^{*}$~FRR is reported at FAR$=0.001$.\\
$^{\dagger}$~Baseline results are taken from the TSCA~\cite{TSCA}.\\
$^{\ddagger}$~Our results are obtained with \emph{encrypted inputs} under TFHE-based secure inference.\\
\end{flushleft}
}
\end{table}

As shown in Table~\ref{tab:overall-one}, FBAN achieves the best AUC on both datasets (98.93 on \textsc{TouchPitfalls} and 99.87 on \textsc{Touchalytics}); moreover, it attains the lowest EER on \textsc{TouchPitfalls} (1.18), while remaining competitive on \textsc{Touchalytics} (EER$=1.65$). 
In addition, FBAN is substantially more efficient than heavy Transformer baselines in terms of FLOPs/Params, while preserving strong authentication accuracy. 

Most importantly, unlike all compared baselines whose results are reported under plaintext inference, our FBAN results are obtained with \emph{encrypted inputs} under TFHE-based secure inference, meaning that the server never observes raw behavioral features in the clear during authentication. 
This demonstrates that FBAN delivers state-of-the-art continuous authentication performance together with a cryptographic privacy guarantee (privacy-in-use) that is absent from existing baselines.

\subsection{Effect of Ensemble Size}
\label{subsec:exp-ens}

We study the impact of the ensemble size $N$ on authentication performance under three deployment pipelines: \textsc{PreTrain} (floating-point), \textsc{QAT} (integer-only quantized), and \textsc{FHE} (TFHE encrypted inference). We vary $N\in\{1,\ldots,7\}$ and report AUC, EER, and FRR at FAR$=0.001$ on both \textsc{TouchPitfalls} and \textsc{Touchalytics}. Table~\ref{tab:ens-size-vert} summarizes the results.

\begin{table}[htbp]
\centering
\caption{Impact of ensemble size $N$ under \textsc{PreTrain}, \textsc{QAT}, and \textsc{FHE} (TFHE).}
\label{tab:ens-size-vert}
\scriptsize
\setlength{\tabcolsep}{2.2pt}
\renewcommand{\arraystretch}{1.03}
\begin{tabular}{@{}c c ccc ccc ccc@{}}
\toprule
\multirow{2}{*}{\textbf{Dataset}} & \multirow{2}{*}{\textbf{$N$}}
& \multicolumn{3}{c}{\textbf{PreTrain}}
& \multicolumn{3}{c}{\textbf{QAT}}
& \multicolumn{3}{c}{\textbf{FHE}} \\
\cmidrule(lr){3-5}\cmidrule(lr){6-8}\cmidrule(lr){9-11}
& & \textbf{AUC}$\uparrow$ & \textbf{FRR}$^{*}\downarrow$ & \textbf{EER}$\downarrow$
  & \textbf{AUC}$\uparrow$ & \textbf{FRR}$^{*}\downarrow$ & \textbf{EER}$\downarrow$
  & \textbf{AUC}$\uparrow$ & \textbf{FRR}$^{*}\downarrow$ & \textbf{EER}$\downarrow$ \\
\midrule

\multirow{7}{*}{\rotatebox{90}{\textsc{TouchPitfalls}}}
& 1 & 98.85 & 15.33 & 1.89 & 98.88 & 13.04 & 1.58 & 98.88 & 13.04 & 1.58 \\
& 2 & 98.89 & 12.94 & 1.69 & 98.92 & 10.93 & 1.35 & 98.91 & 11.23 & 1.42 \\
& 3 & 98.90 & 11.57 & 1.57 & 98.92 & 9.93  & 1.26 & 98.93 & 10.20 & 1.26 \\
& 4 & 98.90 & 12.09 & 1.56 & 98.93 & 9.81  & 1.24 & 98.93 & 10.06 & 1.28 \\
& 5 & 98.91 & 11.69 & 1.49 & 98.93 & 9.48  & 1.16 & 98.93 & 9.73  & 1.19 \\
& 6 & 98.91 & 11.06 & 1.52 & 98.93 & 9.05  & 1.18 & 98.93 & 9.07  & 1.18 \\
& 7 & 98.91 & 11.33 & 1.50 & 98.94 & 9.33  & 1.16 & 98.93 & 9.38  & 1.16 \\
\midrule

\multirow{7}{*}{\rotatebox{90}{\textsc{Touchalytics}}}
& 1 & 99.85 & 5.03 & 1.69 & 99.83 & 5.68 & 1.88 & 99.83 & 5.68 & 1.88 \\
& 2 & 99.87 & 4.87 & 1.58 & 99.85 & 5.45 & 1.76 & 99.85 & 5.45 & 1.76 \\
& 3 & 99.87 & 4.76 & 1.59 & 99.86 & 5.26 & 1.71 & 99.86 & 5.26 & 1.71 \\
& 4 & 99.88 & 4.74 & 1.54 & 99.87 & 5.07 & 1.65 & 99.87 & 5.07 & 1.65 \\
& 5 & 99.88 & 4.73 & 1.56 & 99.87 & 5.18 & 1.62 & 99.87 & 5.18 & 1.62 \\
& 6 & 99.88 & 4.70 & 1.55 & 99.87 & 5.20 & 1.61 & 99.87 & 5.20 & 1.61 \\
& 7 & 99.88 & 4.66 & 1.55 & 99.87 & 5.12 & 1.60 & 99.87 & 5.12 & 1.60 \\
\bottomrule
\end{tabular}

\vspace{0.20em}
{\tiny
\begin{flushleft}
$^{*}$~FRR is reported at FAR$=0.001$.
\end{flushleft}
}
\end{table}

% Table~\ref{tab:ens-size-vert} shows that increasing the ensemble size $N$ generally improves performance when moving from a single member to a moderate ensemble, as ensembling reduces variance and stabilizes the decision boundary under session-level noise and behavioral fluctuations. This trend is most visible in the reduction of FRR and EER for small-to-medium $N$ (e.g., from $N=1$ to $N\approx 4$--$6$), while AUC exhibits smaller but consistent gains. However, the benefit is not monotonic and quickly saturates: beyond a certain point, additional members yield diminishing returns because the ensemble members become increasingly correlated (sharing similar inductive biases and training signals). In some cases, overly large $N$ can even slightly degrade FRR/EER, which we attribute to an ``over-confident'' aggregation effect---when many highly similar models make correlated errors, stacking/averaging may amplify miscalibrated confidence on hard negatives at the fixed FAR operating point, leading to a higher miss rate. 

% Overall, these results suggest a practical accuracy--efficiency balance rather than an ``always larger is better'' rule, and a moderate $N$ provides the most favorable trade-off in our setting.

Table~\ref{tab:ens-size-vert} shows that increasing the ensemble size $N$ typically improves robustness from small to moderate ensembles (e.g., $N=1$ to $N\approx 4$--$6$), mainly reflected by lower FRR/EER, while AUC exhibits smaller but consistent gains. The improvement, however, quickly saturates as ensemble members become more correlated; overly large $N$ may even slightly degrade FRR/EER due to an ``over-confident'' aggregation effect at the fixed operating point (FAR$=0.001$). Overall, a moderate $N$ yields the best accuracy--efficiency balance in our setting.

We further quantify the deployment overhead: one authentication instance takes $561.43$\,ms under plaintext inference, whereas TFHE end-to-end execution takes $73.25$\,s when including encryption, communication, server evaluation, and decryption. Although this latency is intrinsic to bootstrapped FHE, it is acceptable for continuous authentication as a supplementary security layer (e.g., periodic or risk-triggered re-verification) rather than a replacement for the instantaneous primary unlock, allowing the cryptographic cost to be amortized over longer interaction windows.

\paragraph{Scalability discussion.}
Our protocol maintains a per-user template $\mathsf{T}_{\mathsf{uid}}$ comprising $N$ TFHE-compiled circuits and an $N$-dimensional integer stacking-weight vector ($N\le 7$ integers, negligible storage). Each user's ensemble members are independently compiled offline and loaded on demand. For large-scale deployments, standard lifecycle strategies can bound the server-side working set: lazy loading and LRU caching of active user circuits, offline parallel compilation during enrollment, periodic re-enrollment to accommodate behavioral drift, and evaluation-key rotation aligned with organizational policy.

\section{Security Analysis}
\label{sec:security}

We provide a cryptographic security analysis of $\Pi_{\mathrm{FBAN}}$. The analysis covers (i) confidentiality of behavioral features against an honest-but-curious (HbC) server, and (ii) resistance against external malicious adversaries, including replay and user impersonation. Throughout, we reuse the notation established earlier: the client $\mathcal{C}$ holds the TFHE secret key $\mathsf{sk}$; the server $\mathcal{S}$ holds the TFHE evaluation key $\mathsf{ek}$ and the enrolled policy
$\mathsf{T}_{\mathsf{uid}}=(\{\mathcal{C}_{\mathsf{uid},i}\}_{i=1}^{N},\hat{\mathbf{w}}_{\mathsf{uid}},\hat{\tau}_{\mathsf{uid}})$.
At time step $t$, $\mathcal{C}$ quantizes $\mathbf{X}_t$ into $\tilde{\mathbf{X}}_t:=Q_b(\mathsf{Norm}(\mathbf{X}_t))\in\mathbb{Z}_{2^b}^{L\times d}$ and the server samples fresh session identifier $\mathsf{sid}_t$ and input mask $\boldsymbol{\beta}_t\in\mathbb{Z}_{2^b}^{L\times d}$.
In all experiments, TFHE parameters are selected to meet a target security level of at least $128$ bits ($\lambda\ge 128$) using the backend library's parameter-selection routine.

\subsection{TFHE Assumptions}
We rely on the standard guarantees of TFHE:
(i) IND-CPA security of $\mathsf{TFHE.Enc}$, and
(ii) correctness of homomorphic evaluation, i.e., for any circuit $\mathcal{F}$ and message $m$,
\[
\mathsf{TFHE.Dec}\bigl(\mathsf{sk},\,\mathsf{TFHE.Eval}(\mathsf{ek},\mathcal{F},\mathsf{TFHE.Enc}(\mathsf{sk},m))\bigr)=\mathcal{F}(m)
\]
except with negligible probability in $\lambda$.

\subsection{Correctness of Encrypted Authentication}
At time $t$, $\mathcal{S}$ samples $\mathsf{sid}_t$ and $\boldsymbol{\beta}_t \stackrel{\$}{\leftarrow}\mathbb{Z}_{2^b}^{L\times d}$ and sends $(\mathsf{sid}_t,\boldsymbol{\beta}_t)$ to $\mathcal{C}$.
The client forms $\tilde{\mathbf{X}}_t^\star:=\tilde{\mathbf{X}}_t+\boldsymbol{\beta}_t$ and sends
$\mathsf{ct}_{x,t}:=\mathsf{TFHE.Enc}(\mathsf{sk},\tilde{\mathbf{X}}_t^\star)$ to $\mathcal{S}$.
The server performs in-circuit unblinding (subtracting $\boldsymbol{\beta}_t$), evaluates $\{\mathcal{C}_{\mathsf{uid},i}\}_{i=1}^{N}$, and computes the stacked score
\begin{equation}
S_t=\sum_{i=1}^{N}\hat{w}_{\mathsf{uid},i}\cdot s_t^{(i)},\qquad s_t^{(i)}=\mathcal{C}_{\mathsf{uid},i}(\tilde{\mathbf{X}}_t).
\label{eq:sec:stacking}
\end{equation}
By TFHE correctness, the encrypted execution yields the same $S_t$ (and thus the same accept/reject decision) as the corresponding plaintext computation, up to negligible probability.

\subsection{Confidentiality Against an Honest-but-Curious Server}
We capture unavoidable service leakage at time $t$ as
\[
\mathcal{L}_t := (\mathsf{uid},\mathsf{sid}_t, L,d,b,N,\mathsf{Dec}_t),
\]
where $\mathsf{Dec}_t\in\{\textsf{Accept},\textsf{Reject}\}$ is the server decision. Let $\mathsf{View}_{\mathcal{S}}$ denote the full view of $\mathcal{S}$, including $\mathsf{ek}$, $\mathsf{T}_{\mathsf{uid}}$, all server randomness, and all received messages.

\begin{theorem}[Feature Confidentiality w.r.t.\ an HbC Server]
\label{thm:sec-privacy}
Assume TFHE is IND-CPA secure. For any PPT HbC server $\mathcal{S}$, there exists a PPT simulator $\mathsf{Sim}$ such that, for any behavioral stream $\{\tilde{\mathbf{X}}_t\}_t$, the distributions
\[
\mathsf{View}_{\mathcal{S}} \;\approx_c\; \mathsf{Sim}\bigl(\mathsf{ek},\mathsf{T}_{\mathsf{uid}},\{\mathcal{L}_t\}_t\bigr)
\]
are computationally indistinguishable.
\end{theorem}

\noindent\textit{Proof sketch.}
For each time step $t$, the only feature-dependent message received by $\mathcal{S}$ is $\mathsf{ct}_{x,t}=\mathsf{TFHE.Enc}(\mathsf{sk},\tilde{\mathbf{X}}_t+\boldsymbol{\beta}_t)$.
By IND-CPA security (even given $\mathsf{ek}$), $\mathsf{ct}_{x,t}$ is indistinguishable from $\mathsf{TFHE.Enc}(\mathsf{sk},\mathbf{0})$ for a zero tensor $\mathbf{0}$ of the same shape. Replacing real ciphertexts with encryptions of $\mathbf{0}$ yields a standard hybrid argument; applying the (public) PPT procedure $\mathsf{TFHE.Eval}$ preserves indistinguishability. The simulator samples the same server-side randomness and outputs a simulated transcript consistent with the declared leakage $\{\mathcal{L}_t\}_t$.

\subsection{Impersonation Security: Formal Experiment and Bound}
We now formalize resistance to impersonation by an external adversary without $\mathsf{sk}$. Intuitively, $\Pi_{\mathrm{FBAN}}$ requires the client to decrypt the server-produced ciphertext and return a value consistent with the current session; without $\mathsf{sk}$, an impersonator can neither (i) generate fresh well-formed TFHE encryptions, nor (ii) decrypt the server's ciphertext to produce the correct response.

\paragraph{Impersonation experiment.}
Fix a user identity $\mathsf{uid}$. Consider the following experiment $\mathsf{Exp}^{\mathsf{imp}}_{\Pi_{\mathrm{FBAN}}}(\mathcal{A},\lambda)$:

\begin{enumerate}
\item The challenger runs $(\mathsf{sk},\mathsf{ek})\leftarrow \mathsf{TFHE.KeyGen}(1^\lambda)$ and gives $\mathsf{ek}$ and $\mathsf{T}_{\mathsf{uid}}$ to $\mathcal{A}$.
\item $\mathcal{A}$ is given oracle access to an \emph{enrollment oracle} (optional) that returns the public enrollment transcript for $\mathsf{uid}$, and to an \emph{execution oracle} $\mathcal{O}_{\mathsf{auth}}$ that simulates honest sessions with the legitimate client $\mathcal{C}$ for adaptively chosen time steps (i.e., the full transcript of a run of $\Pi_{\mathrm{FBAN}}$ with honest $\mathcal{C}$, including the messages exchanged).\footnote{This is a strong model: the adversary may observe prior authentication transcripts.}
\item Eventually, $\mathcal{A}$ outputs a purported client message strategy for a fresh time step $t^\star$ with a fresh session identifier $\mathsf{sid}_{t^\star}$ (not previously used by $\mathcal{O}_{\mathsf{auth}}$). The challenger runs the server-side algorithm with fresh randomness, interacts with $\mathcal{A}$ as the client, and outputs $1$ iff the server outputs \textsf{Accept}.
\end{enumerate}

Define the impersonation advantage as
\[
\mathsf{Adv}^{\mathsf{imp}}_{\Pi_{\mathrm{FBAN}}}(\mathcal{A},\lambda)
:=\Pr\!\left[\mathsf{Exp}^{\mathsf{imp}}_{\Pi_{\mathrm{FBAN}}}(\mathcal{A},\lambda)=1\right].
\]

\paragraph{Bound (protocol-level impersonation resistance).}
Our interactive design enforces \emph{decrypt-and-return}: the server sends a ciphertext that decrypts to a session-dependent masked value, and the client must return the corresponding plaintext. Without $\mathsf{sk}$, the best possible strategy is to guess a response.

To make this statement precise, we use the \emph{masked decision return} in Algorithm~\ref{alg:protocol} (Eq.~\eqref{eq:fw-blindout}).
At each $t$, the server computes an encrypted decision bit
$d_t:=\mathbb{I}[S_t\ge \hat{\tau}_{\mathsf{uid}}]\in\{0,1\}$ under TFHE, samples a fresh pad $r_t\xleftarrow{\$}\mathbb{Z}_{2^\ell}$, and sends a ciphertext of
$U_t:=r_t+d_t \bmod 2^\ell$ to the client.
The client must decrypt and return $U_t$ in the clear; the server accepts iff $(U_t-r_t)\bmod 2^\ell = 1$.

\begin{theorem}[Impersonation Soundness in a Fresh Session]
\label{thm:impersonation}
Assume (i) TFHE is IND-CPA secure and (ii) $r_t$ is uniform over $\mathbb{Z}_{2^\ell}$ and hidden from the adversary. Then for any PPT adversary $\mathcal{A}$ without access to $\mathsf{sk}$,
\[
\mathsf{Adv}^{\mathsf{imp}}_{\Pi_{\mathrm{FBAN}}}(\mathcal{A},\lambda)
\;\le\;
2^{-\ell} + \mathsf{negl}(\lambda).
\]
\end{theorem}

\noindent\textit{Proof sketch.}
In a fresh session, the server's message to the client is a TFHE ciphertext of $U_t=r_t+d_t$ with a fresh, secret $r_t$.
Any adversary without $\mathsf{sk}$ cannot decrypt this ciphertext; thus its returned value $\hat{U}$ is independent of $r_t$ except with negligible probability (otherwise one could break TFHE security).
The server accepts only if $\hat{U}=r_t+1 \bmod 2^\ell$, which occurs with probability exactly $1/2^\ell$ for a uniform $r_t$.
The negligible term accounts for standard reduction loss and TFHE correctness failure events.

\paragraph{Interpretation.}
Theorem~\ref{thm:impersonation} formalizes that, without $\mathsf{sk}$, a transcript-level impersonator cannot force an \textsf{Accept} outcome in a fresh session except with negligible probability. In practice, choosing $\ell\ge 64$ makes this probability negligible.
\subsection{Replay and Transcript Reuse}
A replay adversary that resends a previously recorded ciphertext $\mathsf{ct}_{x,t_0}$ at time $t\neq t_0$ faces a fresh $(\mathsf{sid}_t,\boldsymbol{\beta}_t,r_t)$.
Even if the adversary replays an old ciphertext, the server's final message in the session is a fresh ciphertext of the masked decision $U_t=r_t+d_t$ (Eq.~\eqref{eq:fw-blindout}), which must be decrypted and returned by the client.
Without $\mathsf{sk}$, the adversary cannot answer this fresh decrypt-and-return challenge, so replay does not yield a valid accepting transcript for a fresh session except with probability bounded by Theorem~\ref{thm:impersonation}. The server additionally rejects reused session identifiers $\mathsf{sid}$.

\subsection{Summary}
Under standard TFHE assumptions, $\Pi_{\mathrm{FBAN}}$ provides privacy-in-use of behavioral features against an HbC server. Moreover, we formalize impersonation resistance via $\mathsf{Exp}^{\mathsf{imp}}$ and show that, without $\mathsf{sk}$, any PPT adversary can succeed only with essentially random-guess probability (plus negligible terms), while replay is prevented by session freshness and the decrypt-and-return structure. All reported experiments correspond to TFHE parameter sets selected at a target security level of at least $128$ bits.

\section{Conclusions}
\label{sec:conclusion}

This paper studied privacy-preserving continuous behavioral authentication and presented \textbf{FBAN}, an FHE-compatible bottleneck-attention network tailored to TFHE-style integer-only computation. By integrating quantization-aware training with homomorphic-friendly attention design, FBAN enables end-to-end encrypted inference while maintaining a lightweight model footprint practical for periodic or risk-triggered re-verification scenarios.

Our experiments on TouchPitfalls and Touchalytics show that FBAN achieves competitive authentication accuracy under a unified evaluation pipeline, and remains effective across \textsc{Pretrain}, \textsc{QAT}, and \textsc{TFHE} settings. From a cryptographic perspective, we formalized feature confidentiality against an honest-but-curious server under standard TFHE assumptions and provided a protocol-level impersonation experiment, showing that an adversary without the TFHE secret key can succeed only with negligible (random-guess) probability.

Overall, FBAN demonstrates that strong continuous-authentication performance can be obtained without exposing sensitive behavioral traces to the server, making it a practical step toward cryptography-grade privacy in behavioral biometrics.

\section*{Acknowledgment}
This work is partially supported by JST BOOST, Japan Grant Number JPMJBS2402, JSPS KAKENHI Grant Number 26H02496 and SECOM Science and Technology Foundation.

% \section*{References}
\bibliographystyle{splncs04}
\bibliography{ref2}

@article{abuhamad2020sensor,
  title   = {Sensor-based continuous authentication of smartphones’ users using behavioral biometrics: A contemporary survey},
  author  = {Abuhamad, Mohammed and Abusnaina, Ahmed and others},
  journal = {IEEE Internet of Things Journal},
  year    = {2020},
  volume  = {8},
  number  = {1},
  pages   = {65--84}
}

@article{stylios2021behavioral,
  title   = {Behavioral biometrics \& continuous user authentication on mobile devices: A survey},
  author  = {Stylios, Ioannis and Kokolakis, Spyros and others},
  journal = {Information Fusion},
  year    = {2021},
  volume  = {66},
  pages   = {76--99}
}

@inproceedings{TSCA,
  title     = {TSCA: Enhancing Smartphone Security with a Touchpoints-Sensitive and Context-Aware Model for Touchscreen-Based Authentication},
  author    = {Wang, Yingjie and Hu, Ruimin},
  booktitle = {2025 28th International Conference on Computer Supported Cooperative Work in Design (CSCWD)},
  year      = {2025},
  pages     = {1937--1942}
}

@article{ayeswarya2024comprehensive,
  title   = {A comprehensive review on secure biometric-based continuous authentication and user profiling},
  author  = {Ayeswarya, S and Singh, K John},
  journal = {IEEE Access},
  year    = {2024},
  volume  = {12},
  pages   = {82996--83021}
}

@article{baig2021security,
  title   = {Security, privacy, and usability in continuous authentication: A survey},
  author  = {Baig, Ahmed Fraz and Eskeland, Sigurd},
  journal = {Sensors},
  year    = {2021},
  volume  = {21},
  number  = {17},
  pages   = {5967}
}

@inproceedings{yang2012novel,
  title     = {A novel biometric authentication scheme with privacy preserving},
  author    = {Yang, Dexin and Xu, Baolin and others},
  booktitle = {2012 Eighth International Conference on Computational Intelligence and Security},
  year      = {2012},
  pages     = {452--456}
}

@inproceedings{xiong2024privacy,
  title        = {A Privacy-Preserving Computer-Aided Diagnosis Framework for Medical Applications Using Federated Learning and Homomorphic Encryption},
  author       = {Xiong, Jichao and Chen, Jiageng and Liu, Hui and Zhou, Guangyou and Cui, Jianqun and Lin, Junyu and Jiao, Dian},
  booktitle    = {International Conference on Attacks and Defenses for Internet-of-Things},
  pages        = {90--106},
  year         = {2024},
  organization = {Springer}
}

@article{xiong2024enhancing,
  title   = {Enhancing privacy-preserving machine learning with self-learnable activation functions in fully homomorphic encryption},
  author  = {Xiong, Jichao and Chen, Jiageng and others},
  journal = {Journal of Information Security and Applications},
  year    = {2024},
  volume  = {86},
  pages   = {103887}
}

@article{chillotti2020tfhe,
  title   = {TFHE: fast fully homomorphic encryption over the torus},
  author  = {Chillotti, Ilaria and Gama, Nicolas and others},
  journal = {Journal of Cryptology},
  year    = {2020},
  volume  = {33},
  number  = {1},
  pages   = {34--91}
}

@ARTICLE{Touchalytics,
  author  = {Frank, Mario and Biedert, Ralf and others},
  title   = {Touchalytics: On the Applicability of Touchscreen Input as a Behavioral Biometric for Continuous Authentication},
  journal = {IEEE Transactions on Information Forensics and Security},
  year    = {2013},
  volume  = {8},
  number  = {1},
  pages   = {136--148}
}

@article{lin2025lightweight,
  title   = {A Lightweight Privacy-Preserving Federated Learning Framework for Heterogeneity-Resilient Skin Cancer Diagnosis},
  author  = {Lin, Junyu and Chen, Jiageng and others},
  journal = {IEEE Journal of Biomedical and Health Informatics},
  year    = {2025}
}

@misc{ConcreteML,
  title  = {Concrete {ML}: a Privacy-Preserving Machine Learning Library using Fully Homomorphic Encryption for Data Scientists},
  author = {Zama},
  year   = {2022}
}

@article{sagi2018ensemble,
  title   = {Ensemble learning: A survey},
  author  = {Sagi, Omer and Rokach, Lior},
  journal = {Wiley interdisciplinary reviews: data mining and knowledge discovery},
  year    = {2018},
  volume  = {8},
  number  = {4},
  pages   = {e1249}
}

@book{gentry2009fully,
  title  = {A fully homomorphic encryption scheme},
  author = {Gentry, Craig},
  year   = {2009},
  publisher = {Stanford university}
}

@article{sitova2015hmog,
  title   = {HMOG: New behavioral biometric features for continuous authentication of smartphone users},
  author  = {Sitov{\'a}, Zde{\v{n}}ka and {\v{S}}ed{\v{e}}nka, Jaroslav and others},
  journal = {IEEE Transactions on Information Forensics and Security},
  year    = {2015},
  volume  = {11},
  number  = {5},
  pages   = {877--892}
}

@article{dilshan2023behaveformer,
  title   = {Behaveformer: a framework with spatio-temporal dual attention transformers for imu enhanced keystroke dynamics},
  author  = {Dilshan, S and Sanuja, T and others},
  journal = {Proceedings in Computer Science, Cryptography and Security},
  year    = {2023}
}

@article{ometov2018multi,
  title   = {Multi-factor authentication: A survey},
  author  = {Ometov, Aleksandr and Bezzateev, Sergey and others},
  journal = {Cryptography},
  year    = {2018},
  volume  = {2},
  number  = {1},
  pages   = {1}
}

@article{chen2022x,
  title   = {The-x: Privacy-preserving transformer inference with homomorphic encryption},
  author  = {Chen, Tianyu and Bao, Hangbo and others},
  journal = {arXiv preprint arXiv:2206.00216},
  year    = {2022}
}

@inproceedings{georgiev2022techniques,
  title     = {Techniques for continuous touch-based authentication},
  author    = {Georgiev, Martin and Eberz, Simon and others},
  booktitle = {International Conference on Information Security Practice and Experience},
  year      = {2022},
  pages     = {409--431}
}

@inproceedings{bai2023cryptomask,
  title     = {CryptoMask: Privacy-preserving face recognition},
  author    = {Bai, Jianli and Zhang, Xiaowu and others},
  booktitle = {International Conference on Information and Communications Security},
  year      = {2023},
  pages     = {333--350}
}

@inproceedings{gilad2016cryptonets,
  title     = {Cryptonets: Applying neural networks to encrypted data with high throughput and accuracy},
  author    = {Gilad-Bachrach, Ran and Dowlin, Nathan and others},
  booktitle = {International conference on machine learning},
  year      = {2016},
  pages     = {201--210}
}

@article{zhang2024secure,
  title   = {Secure transformer inference made non-interactive},
  author  = {Zhang, Jiawen and Yang, Xinpeng and others},
  journal = {Cryptology ePrint Archive},
  year    = {2024}
}

@techreport{park2024powerformer,
  title       = {Powerformer: Efficient privacy-preserving transformer with batch rectifier-power max function and optimized homomorphic attention},
  author      = {Park, Dongjin and Lee, Eunsang and others},
  institution = {Cryptology ePrint Archive},
  year        = {2024}
}

@inproceedings{AsiaCCS2022,
  title     = {Common evaluation pitfalls in touch-based authentication systems},
  author    = {Georgiev, Martin and Eberz, Simon and others},
  booktitle = {Proceedings of the 2022 ACM on Asia Conference on Computer and Communications Security},
  year      = {2022},
  pages     = {1049--1063}
}

@article{SVMDA,
  title   = {DAKOTA: Sensor and touch screen-based continuous authentication on a mobile banking application},
  author  = {Incel, {\"O}zlem Durmaz and G{\"u}nay, Se{\c{c}}il and others},
  journal = {IEEE Access},
  year    = {2021},
  volume  = {9},
  pages   = {38943--38960}
}

@article{PatchTST,
  title   = {A Time Series is Worth 64Words: Long-term Forecasting with Transformers},
  author  = {Nie, Y},
  journal = {arXiv preprint arXiv:2211.14730},
  year    = {2022}
}

@article{Timesnet,
  title   = {Timesnet: Temporal 2d-variation modeling for general time series analysis},
  author  = {Wu, Haixu and Hu, Tengge and others},
  journal = {arXiv preprint arXiv:2210.02186},
  year    = {2022}
}

@article{OSCNN,
  title   = {Omni-scale cnns: a simple and effective kernel size configuration for time series classification},
  author  = {Tang, Wensi and Long, Guodong and others},
  journal = {arXiv preprint arXiv:2002.10061},
  year    = {2020}
}

@inproceedings{CKKS,
  title={Homomorphic encryption for arithmetic of approximate numbers},
  author={Cheon, Jung Hee and Kim, Andrey and Kim, Miran and Song, Yongsoo},
  booktitle={International conference on the theory and application of cryptology and information security},
  pages={409--437},
  year={2017},
  organization={Springer}
}

@inproceedings{xiong2025stresssentry,
  title={StressSentry-FHE: A Transformer-Based Privacy-Preserving Framework for Stress Detection Using Quantized Attention},
  author={Xiong, Jichao and Chen, Jiageng and others},
  booktitle={International Conference on Knowledge Science, Engineering and Management},
  pages={300--311},
  year={2025},
  organization={Springer}
}

@article{yao2026strokepin,
  title={StrokePIN: Enhancing PIN Authentication with Keystroke Dynamics for Mobile Devices},
  author={Yao, Jingyu and Wang, Ding},
  journal={IEEE Transactions on Dependable and Secure Computing},
  year={2026},
  publisher={IEEE}
}

@inproceedings{zhu2023soundlock,
  title={SoundLock: A Novel User Authentication Scheme for VR Devices Using Auditory-Pupillary Response},
  author={Zhu, Huadi and Xiao, Mingyan and others},
  booktitle={NDSS},
  year={2023}
}

\end{document}